\documentclass[journal,twoside,web]{ieeecolor}
\usepackage{generic}  %

\let\labelindent\relax
\usepackage{amsfonts,amsmath, amssymb, mathtools, xcolor, hyperref, graphicx,algorithm, bbm, bm, algpseudocode, booktabs, float, optidef,subcaption,cite, multirow,enumitem, tikz} 
 \usetikzlibrary{arrows.meta,bending}

\usepackage{flushend}
\usepackage{amsthm}
\usepackage{cancel}
\usepackage[normalem]{ulem}
\usepackage{placeins}
\hypersetup{hidelinks}
\usepackage{lcsys}   
\theoremstyle{plain} %
\newtheorem{theorem}{Theorem}

\newtheorem{proposition}{Proposition}
\newtheorem{corollary}{Corollary}
\usepackage{wrapfig}
\theoremstyle{definition} %
\newtheorem{definition}{Definition}

\theoremstyle{remark}

\newtheorem*{remark*}{Remark}
\title{Topology-Aware Congestion Pricing: \\ Demand-Robust Routing using the Forman–Ricci Curvature
}

\author{\centerline{Sheryl Paul$^{1}$, Samuel J. Williams$^{1}$, Kexin Wang$^{1}$, Xin Qin$^{2}$,}
\centerline{Ruolin Li$^{1}$, Jyotirmoy V. Deshmukh$^{1}$}%
\thanks{$^{1}$ University of Southern California,
Los Angeles, CA 90089, USA.
{\tt\small \{sherylpa, kwang255, samwilliams, ruolinli, jdeshmuk\}@usc.edu}}%
\thanks{$^{2}$ California State University, Long Beach,
Long Beach, CA 90840, USA.
{\tt\small xin.qin@csulb.edu}}%
\thanks{\copyright~2026 IEEE. Personal use of this material is permitted. Permission from IEEE must be obtained for all other uses, in any current or future media, including reprinting/republishing this material for advertising or promotional purposes, creating new collective works, for resale or redistribution to servers or lists, or reuse of any copyrighted component of this work in other works. Accepted for publication in IEEE Control Systems Letters.}%
}

\newcommand{\graph}{G}
\newcommand{\nodes}{V}
\newcommand{\edges}{\mathcal{E}}
\newcommand{\numnodes}{n}
\newcommand{\numedges}{m}

\newcommand{\ods}{\mathcal{W}}
\newcommand{\od}{w}
\newcommand{\origin}{s}
\newcommand{\dest}{t}
\newcommand{\demand}{d}

\newcommand{\ppaths}{\mathcal{P}}
\newcommand{\ppath}{p}
\newcommand{\flowvector}{x}
\newcommand{\flowvectors}{\mathcal{X}}

\newcommand{\edgelatency}{\ell_e}

\newcommand{\syscost}{\mathcal{J}}

\newcommand{\beck}{\Phi}

\newcommand{\regparam}{\lambda}

\providecommand{\rev}[1]{#1}
\newcommand{\poa}{\mathrm{PoA}}

\newcommand{\formcurve}{F_e}

\definecolor{bestcell}{RGB}{212,237,218}

\newcommand{\penweight}{\mu_e}           %
\newcommand{\penweightvec}{\mu}          %
\newcommand{\latenop}{\mathbf{C}}        %

\begin{document}

\maketitle
\thispagestyle{empty}
\pagestyle{empty}

\begin{abstract}
Congestion pricing is widely used to improve transportation network efficiency where individual decentralized route choices  can lead to system-level inefficiencies. Pigouvian tolls achieve the system optimum but depend on demand
and can lose effectiveness under demand shifts, particularly on
structurally critical bottleneck edges.
We propose a topology-aware regularization of the Beckmann potential governing the (Wardrop) equilibrium by augmenting precomputed Pigouvian tolls with an offline structural penalty derived from Forman--Ricci curvature (FRC), which captures bottleneck structures.
We prove that the resulting equilibrium is well-defined, 
that stronger regularization reduces \rev{the penalty-weighted aggregate} flow on penalized edges,
and that efficiency loss depends on how closely the structural penalty approximates optimal Pigouvian tolls.
Experiments on eight real-world networks \rev{show up to 29\% closer tracking of
optimal bottleneck flow while changing the worst-case cost by under 0.25\%, with a penalty computed up to
two orders of magnitude faster than edge betweenness, making the correction} a
practical, robust complement to classical congestion pricing.
\end{abstract}


\section{INTRODUCTION}

\IEEEPARstart{L}{arge}-scale infrastructure systems such as transportation and supply
chain networks rely on decentralized, individual decision-making based on
local objectives. In transportation networks, selfish
route choices produce a Wardrop equilibrium that can be inefficient relative
to the social optimum (SO), as quantified by the Price of Anarchy
\cite{roughgarden2002bad}. These systems are also vulnerable to demand shifts,
which can disproportionately congest structurally critical edges. This
motivates interventions that are both efficient and robust.

\noindent Congestion pricing controls this equilibrium by charging users for the delays
they impose on others \cite{pigou1920economics, beckmann1956studies}.
Pigouvian tolls calibrated to a known demand align selfish routing with the SO
\cite{roughgarden2002bad, roughgarden2003price, cole2003pricing}, but must
generally be recomputed when demand changes. Consequently, tolls calibrated
at a base demand may lose effectiveness under daily variation, long-term
growth, or unexpected events \cite{friesz2011dynamiccongestion,
lo2002rstochastic}. 

\noindent From a control perspective, congestion pricing designs an input
that shapes the equilibrium of a decentralized system, and tolls should be
efficient under nominal conditions yet robust to disturbances such as demand
variation. Robust and bilevel formulations
\cite{zhang2022drso, schmidt2024robust} require demand-specific optimization,
and adaptive methods rely on online recomputation \cite{li2023adaptive}.
\rev{Computing such tolls is relatively inexpensive, but shifted demand is
often unobserved and tolls cannot be revised frequently.} We therefore seek
an offline structural correction that remains fixed under demand variation.

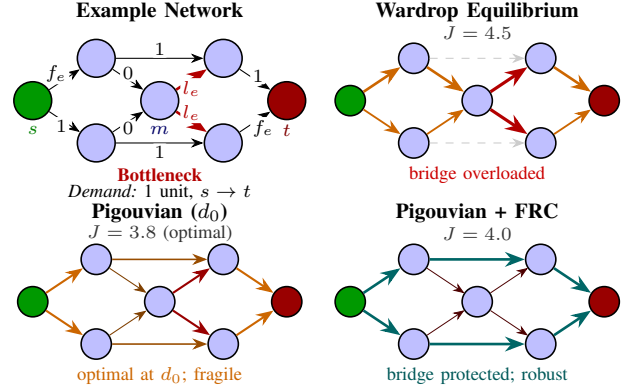
\begin{figure}[!t]

\centering
\begin{tikzpicture}[
  >=Stealth,
  nd/.style={circle, draw=black, semithick, minimum size=5mm, inner sep=0pt},
  src/.style={nd, fill=green!60!black},
  dst/.style={nd, fill=red!60!black},
  mid/.style={nd, fill=blue!25},
  lbl/.style={fill=white, inner sep=1pt, font=\scriptsize, rounded corners=1pt},
  scale=0.28,
]

\node[font=\footnotesize\bfseries, transform shape=false] at (6, 4.2)
  {Example Network};

\node[src] (s) at (0, 0) {};
\node[mid] (a) at (3, 2) {};
\node[mid] (b) at (3, -2) {};
\node[mid] (m) at (6, 0) {};
\node[mid] (c) at (9, 2) {};
\node[mid] (d) at (9, -2) {};
\node[dst] (t) at (12, 0) {};

\node[font=\scriptsize\bfseries, text=green!50!black, transform shape=false] at (0, -1.4) {$s$};
\node[font=\scriptsize\bfseries, text=red!50!black, transform shape=false] at (12, -1.4) {$t$};
\node[font=\scriptsize\bfseries, text=blue!40!black, transform shape=false] at (6, -1.4) {$m$};

\draw[thin, ->] (s) -- (a) node[lbl, midway, above left=-2pt] {$f_e$};
\draw[thin, ->] (a) -- (c) node[lbl, midway, above] {$1$};
\draw[thin, ->] (c) -- (t) node[lbl, midway, above right=-2pt] {$1$};
\draw[thin, ->] (s) -- (b) node[lbl, midway, below left=-2pt] {$1$};
\draw[thin, ->] (b) -- (d) node[lbl, midway, below] {$1$};
\draw[thin, ->] (d) -- (t) node[lbl, midway, below right=-2pt] {$f_e$};
\draw[thin, ->] (a) -- (m) node[lbl, midway, above=-1pt] {$0$};
\draw[thin, ->] (b) -- (m) node[lbl, midway, below=-1pt] {$0$};
\draw[red!70!black, line width=1pt, ->] (m) -- (c)
  node[lbl, midway, below=-1pt, text=red!70!black, font=\scriptsize\bfseries] {$l_e$};
\draw[red!70!black, line width=1pt, ->] (m) -- (d)
  node[lbl, midway, above=-1pt, text=red!70!black, font=\scriptsize\bfseries] {$l_e$};

\node[font=\scriptsize\bfseries, text=red!70!black, transform shape=false]
  at (6, -3.5) {Bottleneck};
\node[font=\scriptsize, transform shape=false]
  at (6, -4.3) {\textit{Demand:} 1 unit, $s \to t$};

\def\rxoff{15}
\node[font=\footnotesize\bfseries, transform shape=false] at (\rxoff+6, 4.2)
  {Wardrop Equilibrium};
\node[font=\scriptsize, text=black!80, transform shape=false] at (\rxoff+6, 3.2)
  {$J = 4.5$};

\node[src, minimum size=4mm] (ws) at (\rxoff+0, 0) {};
\node[mid, minimum size=4mm] (wa) at (\rxoff+3, 2) {};
\node[mid, minimum size=4mm] (wb) at (\rxoff+3, -2) {};
\node[mid, minimum size=4mm] (wm) at (\rxoff+6, 0) {};
\node[mid, minimum size=4mm] (wc) at (\rxoff+9, 2) {};
\node[mid, minimum size=4mm] (wd) at (\rxoff+9, -2) {};
\node[dst, minimum size=4mm] (wt) at (\rxoff+12, 0) {};

\draw[gray!40, thin, dashed, ->] (wa) -- (wc);
\draw[gray!40, thin, dashed, ->] (wb) -- (wd);
\draw[orange!80!black, line width=1pt, ->] (ws) -- (wa);
\draw[orange!80!black, line width=0.7pt, ->] (ws) -- (wb);
\draw[orange!80!black, line width=1pt, ->] (wa) -- (wm);
\draw[orange!80!black, line width=0.7pt, ->] (wb) -- (wm);
\draw[red!70!black, line width=1.2pt, ->] (wm) -- (wc);
\draw[red!70!black, line width=1.2pt, ->] (wm) -- (wd);
\draw[orange!80!black, line width=1pt, ->] (wc) -- (wt);
\draw[orange!80!black, line width=0.7pt, ->] (wd) -- (wt);

\node[font=\scriptsize, text=red!80!black, transform shape=false]
  at (\rxoff+6, -3.5) {bridge overloaded};

\def\byoff{-9.5}
\node[font=\footnotesize\bfseries, transform shape=false] at (6, \byoff+4.2)
  {Pigouvian ($d_0$)};
\node[font=\scriptsize, text=black!80, transform shape=false] at (6, \byoff+3.2)
  {$J = 3.8$ (optimal)};

\node[src, minimum size=4mm] (ps) at (0, \byoff) {};
\node[mid, minimum size=4mm] (pa) at (3, \byoff+2) {};
\node[mid, minimum size=4mm] (pb) at (3, \byoff-2) {};
\node[mid, minimum size=4mm] (pm) at (6, \byoff) {};
\node[mid, minimum size=4mm] (pc) at (9, \byoff+2) {};
\node[mid, minimum size=4mm] (pd) at (9, \byoff-2) {};
\node[dst, minimum size=4mm] (pt) at (12, \byoff) {};

\draw[orange!80!black, line width=0.8pt, ->] (ps) -- (pa);
\draw[orange!80!black, line width=0.8pt, ->] (ps) -- (pb);
\draw[orange!60!black, line width=0.6pt, ->] (pa) -- (pc);
\draw[orange!60!black, line width=0.6pt, ->] (pb) -- (pd);
\draw[orange!50!black, line width=0.4pt, ->] (pa) -- (pm);
\draw[orange!50!black, line width=0.4pt, ->] (pb) -- (pm);
\draw[red!60!black, line width=0.8pt, ->] (pm) -- (pc);
\draw[red!60!black, line width=0.8pt, ->] (pm) -- (pd);
\draw[orange!80!black, line width=0.8pt, ->] (pc) -- (pt);
\draw[orange!80!black, line width=0.8pt, ->] (pd) -- (pt);

\node[font=\scriptsize, text=orange!80!black, transform shape=false]
  at (6, \byoff-3.5) {optimal at $d_0$; fragile};

\node[font=\footnotesize\bfseries, transform shape=false] at (\rxoff+6, \byoff+4.2)
  {Pigouvian + FRC};
\node[font=\scriptsize, text=black!80, transform shape=false] at (\rxoff+6, \byoff+3.2)
  {$J = 4.0$};

\node[src, minimum size=4mm] (fs) at (\rxoff+0, \byoff) {};
\node[mid, minimum size=4mm] (fa) at (\rxoff+3, \byoff+2) {};
\node[mid, minimum size=4mm] (fb) at (\rxoff+3, \byoff-2) {};
\node[mid, minimum size=4mm] (fm) at (\rxoff+6, \byoff) {};
\node[mid, minimum size=4mm] (fc) at (\rxoff+9, \byoff+2) {};
\node[mid, minimum size=4mm] (fd) at (\rxoff+9, \byoff-2) {};
\node[dst, minimum size=4mm] (ft) at (\rxoff+12, \byoff) {};

\draw[teal!80!black, line width=1pt, ->] (fs) -- (fa);
\draw[teal!80!black, line width=1pt, ->] (fs) -- (fb);
\draw[teal!80!black, line width=1pt, ->] (fa) -- (fc);
\draw[teal!80!black, line width=1pt, ->] (fb) -- (fd);
\draw[teal!80!black, line width=1pt, ->] (fc) -- (ft);
\draw[teal!80!black, line width=1pt, ->] (fd) -- (ft);

\draw[red!30!black, line width=0.2pt, ->] (fa) -- (fm);
\draw[red!30!black, line width=0.2pt, ->] (fb) -- (fm);
\draw[red!30!black, line width=0.2pt, ->] (fm) -- (fc);
\draw[red!30!black, line width=0.2pt, ->] (fm) -- (fd);

\node[font=\scriptsize, text=teal!70!black, transform shape=false]
  at (\rxoff+6, \byoff-3.5) {bridge protected; robust};
\end{tikzpicture}
\caption{%
\small
\textit{\textbf{Top-left:} Example network with bottleneck hub~$m$; red edges:
 steep latency $\ell\!=\!2f_e$, negative FRC.
\textbf{Top-right:} WE concentrates flow on bottleneck. 
\textbf{Bottom-left:} Pigouvian at $d_0$ is optimal but fragile to demand shifts.
\textbf{Bottom-right:} Pig.+FRC diverts flow to direct paths (teal) at
small cost for structural protection.}}
\label{fig:motivating}
\end{figure}

\noindent Another motivation is that edges connecting weakly linked regions can be
structurally fragile.
Demand increases
along such corridors can therefore concentrate congestion even when tolls are
optimal at the base demand. We address this vulnerability with a \emph{proactive
structural penalty} that discourages bottleneck overuse before congestion
materializes.

\noindent Curvature provides a natural way to quantify such structural
fragility: on graphs, it scores how well each edge is supported by its local
neighborhood. Forman--Ricci curvature (FRC)
\cite{forman2003bochner, sreejith2016forman} is a local, cheaply computed
instance; \rev{strongly negative curvature marks edges with high endpoint
degree and few short detours}, whereas higher curvature indicates local
redundancy. FRC has been used for community detection \cite{ni2019community},
network robustness \cite{samal2018comparative, weber2017forman}, and
bottleneck detection in graph neural networks
\cite{topping2021oversquashing}. To our knowledge, this work first connects
curvature with routing-game mechanism design.

\noindent We incorporate a \emph{curvature-regularized Beckmann potential} by adding a linear FRC penalty to the Beckmann potential,
equivalently shifting edge costs while preserving the equilibrium
structure. We
combine this penalty with Pigouvian tolls calibrated at a base demand,
producing a \emph{combined toll} designed to remain effective under demand
shifts (Fig.~\ref{fig:motivating}).

\noindent\textbf{Contributions.}
Our main contributions are:
\textit{(1) Well-posedness:} We characterize the regularized equilibrium as a WE under shifted costs.
\textit{(2) Sensitivity:} We establish Hölder sensitivity, Lipschitz sensitivity under positive flow, and non-increasing penalty-weighted flow.
\textit{(3) Efficiency:} We bound social-cost loss through the Pigouvian-toll approximation error.
\textit{(4) Robustness:} We show that alignment with demand-induced toll changes improves the certified bound.
\textit{(5) Validation:} Experiments on eight networks under targeted and uniform stress demonstrate improved bottleneck protection, near-optimal cost, and faster computation.

\section{PRELIMINARIES}

\subsection{Traffic Networks and Routing Games}
\noindent A traffic network is a directed graph
$\graph=(\nodes,\edges)$ with $\numnodes=|\nodes|$ nodes and
$\numedges=|\edges|$ edges. Let
$\ods\subseteq\nodes\times\nodes$ be the set of origin--destination (OD)
pairs, with demand $\demand_\od>0$ for each
$\od=(\origin,\dest)\in\ods$. Let $\ppaths_\od$ be the set of simple
directed paths connecting $\origin$ to $\dest$, and set
$\ppaths=\bigcup_{\od\in\ods}\ppaths_\od$.
A path-flow vector
$\flowvector=(\flowvector_\ppath)_{\ppath\in\ppaths}$ is feasible if
\(
\flowvector_\ppath\geq0,
\ \text{and }
\sum_{\ppath\in\ppaths_\od}\flowvector_\ppath=\demand_\od
\quad\forall\,\od\in\ods.
\)
We denote the feasible path-flow set by $\flowvectors$. The induced edge
flow is
\(
f_e(\flowvector)
:=\sum_{\ppath\in\ppaths:\,e\in\ppath}\flowvector_\ppath,
\ \ 
f(\flowvector):=(f_e(\flowvector))_{e\in\edges},
\)
and $\mathcal{E}(\flowvectors)$ denotes the set of feasible edge flows. We write $f_e$ when the dependence on $\flowvector$ is clear.
Each edge $e$ has a continuous, nondecreasing latency function
$\edgelatency:\mathbb{R}_+\to\mathbb{R}_+$. In the experiments we use the
Bureau of Public Roads (BPR) model
\(
\ell_e(f_e)
=t_e^0\left[1+b_e\left({f_e}/{c_e}\right)^{p}\right],
\) \label{def:bpr}
where $t_e^0>0$ is the free-flow travel time, $c_e>0$ is capacity,
$b_e>0$, and typically $p=4$. The cost of path $\ppath$ is
\(
J_\ppath(\flowvector)
:=\sum_{e\in\ppath}\ell_e\!\left(f_e(\flowvector)\right).
\)

\subsection{Wardrop Equilibrium (WE) and Beckmann Potential}
\begin{definition}[Wardrop equilibrium]
A feasible path flow $\flowvector^{\mathrm{WE}}\in\flowvectors$ is a
\emph{Wardrop equilibrium} (WE) if, for every $\od\in\ods$ and
$\ppath,q\in\ppaths_\od$,
\begin{equation}
\flowvector^{\mathrm{WE}}_\ppath>0
\quad\Longrightarrow\quad
J_\ppath(\flowvector^{\mathrm{WE}})
\leq J_q(\flowvector^{\mathrm{WE}}).
    \label{eq:wardrop}
\end{equation}
\end{definition}
Thus, every used path has minimum cost within its OD pair. For separable
nondecreasing latencies, WE path flows are precisely the minimizers of the
Beckmann potential \cite{beckmann1956studies}
\(
\beck(\flowvector)
:=\sum_{e\in\edges}
\int_0^{f_e(\flowvector)}\ell_e(s)\,\mathrm{d}s
\)
over $\flowvectors$. Indeed,
$\partial\beck/\partial\flowvector_\ppath=J_\ppath(\flowvector)$, so its
optimality conditions coincide with~\eqref{eq:wardrop}. For BPR latencies,
the edge term is
\(
t_e^0\left[
f_e+{(b_e f_e^{p+1})}/{((p+1)c_e^p)}
\right].
\)
Under strictly increasing latencies, the equilibrium edge flow is unique,
although its path-flow decomposition need not be. We write
\(
\latenop(f):=(\ell_e(f_e))_{e\in\edges}
\)
for the latency operator.

\subsection{Social Cost, Social Optimum, and Price of Anarchy}
\begin{definition}[Social Cost and Optimum]
    The social cost of an edge flow $f\in\mathcal{E}(\flowvectors)$ is
\(
\syscost(f):=\sum_{e\in\edges}f_e\ell_e(f_e).
\)
A socially optimal path flow and its induced edge flow are denoted by
\(
\flowvector^{\mathrm{SO}}
\in\arg\min_{\flowvector\in\flowvectors}
\syscost(f(\flowvector)),
\text{ and }
f^{\mathrm{SO}}:=f(\flowvector^{\mathrm{SO}})
\).
\end{definition}
For a WE edge flow $f^{\mathrm{WE}}$, the Price of Anarchy is
\(
\poa:=\frac{\syscost(f^{\mathrm{WE}})}
{\syscost(f^{\mathrm{SO}})}
\geq1.
\)
This gap arises because drivers minimize personal travel time
without accounting for delays imposed on others\rev{; the WE is the pure Nash
equilibrium of the nonatomic routing game with fixed demand}.\footnote{\rev{The PoA is at most $4/3$ for affine
latencies~\cite{roughgarden2002bad} and about $2.151$ for degree-4 polynomials
such as BPR~\cite{roughgarden2003price}; on our networks it is $1.014$--$1.094$
at base demand and at most $1.20$ under the shifted demands.}}
\begin{definition}[Pigouvian Toll]
For differentiable separable latencies and fixed demand,
the Pigouvian toll associated with $f^{\mathrm{SO}}$ is
\(
\tau_e^\star
:=f_e^{\mathrm{SO}}\ell'_e(f_e^{\mathrm{SO}}).
\)
The perceived cost $\ell_e(f_e)+\tau_e^\star$ implements
$f^{\mathrm{SO}}$ as a tolled WE~\cite{beckmann1956studies}, i.e., selfish route choice coincides with the socially optimal flow.
\end{definition}

\section{Problem Statement and Solution Approach}
\label{sec:model}
 
\noindent \textbf{Problem Statement.}
We view congestion pricing as a control problem in which an input toll vector shapes the equilibrium of a decentralized routing
system under demand disturbances.
\rev{Our objective is robustness to demand shifts without recomputing
tolls. Let $d_0$ be the base demand at which tolls are calibrated,
$\mathcal D$ a set of shifted demands around $d_0$, and $f^\tau(d)$ the
tolled WE at demand $d$ under a toll $\tau$ fixed at $d_0$. We measure
robustness by the worst-case efficiency ratio
$\eta_{\mathcal D}(\tau):=\max_{d\in\mathcal D}
\syscost(f^\tau(d))/\syscost(f^{\mathrm{SO}}(d))\ge1$ and seek tolls with
small $\eta_{\mathcal D}$. Additionally, we want to reduce the flow on bottleneck edges.}
To this end,
we augment the base-demand Pigouvian toll
$\tau^0:=\tau^\star(d_0)$ with a fixed structural penalty derived
from Forman--Ricci curvature (FRC).

\noindent\textbf{Curvature-Regularized Routing Model.}
For an unweighted directed graph, we compute the \emph{Forman--Ricci
curvature} (FRC) \cite{forman2003bochner,sreejith2016forman} on its simple
undirected projection $\graph^{\mathrm u}$, obtained by replacing each
directed or reciprocal node pair by one undirected edge. 
\rev{We use the triangle-augmented FRC~\cite{samal2018comparative}
$\formcurve:=4-\deg(u)-\deg(v)+3T_e$, with degrees and the number $T_e$ of
triangles containing $e$ taken in $\graph^{\mathrm u}$.}\footnote{\rev{$\formcurve$ depends only on edges incident to $u$ or $v$ and is
computable in $O(|\edges|\Delta_{\max})$ time. It is invariant to
structure beyond this neighborhood, so it cannot detect long-range bridges or cuts.}}
The calculation uses graph topology and requires no traffic data.
\noindent\textbf{Base-demand AON loading.}
For each OD pair, let
\(
\ppaths_\od^0
:=\arg\min_{\ppath\in\ppaths_\od}
\sum_{e\in\ppath}t_e^0
\)
be its set of free-flow shortest paths. Choose coefficients
$\alpha_{\od,\ppath}\geq0$ satisfying
$\sum_{\ppath\in\ppaths_\od^0}\alpha_{\od,\ppath}=1$. The all-or-nothing
(AON) edge flow at $d_0$ is
$f_e^{\mathrm{AON}}
:=\sum_{\od\in\ods}\demand_\od
\sum_{\ppath\in\ppaths_\od^0}
\alpha_{\od,\ppath}\,\mathbf{1}\{e\in\ppath\}.$
A deterministic shortest-path tie rule corresponds to selecting one
coefficient equal to one for each OD pair.

\noindent \textbf{Design of the penalty weight. }
\noindent For each edge, we design the \emph{curvature penalty weight}:
\setlength{\abovedisplayskip}{3pt plus 1pt minus 1pt}
\setlength{\belowdisplayskip}{3pt plus 1pt minus 1pt}
\setlength{\abovedisplayshortskip}{0pt plus 1pt}
\setlength{\belowdisplayshortskip}{2pt plus 1pt minus 1pt}
\begin{equation}
  \penweight
  := (\max\{0, -\formcurve\})^2 \cdot \ell_e'(f_e^{\mathrm{AON}})
  \;\geq 0, \forall e \in \edges
  \label{eq:penweight}
\end{equation}
It
combines the following components, each serving a distinct role.
 \emph{(1) Thresholding}: $\max\{0, -\formcurve\}$ restricts the penalty to negatively curved edges (interpreted as structural bottlenecks), while squaring emphasizes strongly negative curvature.
\emph{(2) Latency derivative scaling}: Multiplying by $\edgelatency'(f_e^{\mathrm{AON}})$ scales the penalty by congestion
sensitivity.
 \footnote{$\penweightvec$ is computed once at a base demand $d_0$ and then held fixed as demand varies. In contrast to Pigouvian tolls, which require recomputation via equilibrium solves, $\penweightvec$ depends only on a single shortest-path pass and thus provides a reusable structural correction under demand shifts.}

\begin{definition}[Regularized Beckmann potential]
\label{def:regpot}
Let $\mu=(\mu_e)_{e\in\edges}\in\mathbb{R}_{\ge0}^{|\edges|}$
be a fixed edge-penalty vector. For $\lambda\ge0$:
$\Phi_\lambda(x)
:=\sum_{e\in\edges}\int_0^{f_e(x)}
   [\ell_e(s)+\lambda\mu_e]\,\mathrm{d}s =\Phi(x)+\lambda\langle\mu,f(x)\rangle$. 
\end{definition}
The shifted edge costs are
$\tilde{\ell}_e^\lambda(z):=\ell_e(z)+\lambda\mu_e$,
with path costs
$\tilde{J}_p^\lambda(x):=J_p(x)+\lambda\sum_{e\in p}\mu_e$
and operator
$\tilde{\mathbf C}_\lambda(f):=\mathbf C(f)+\lambda\mu$.

\noindent Because $\lambda\mu_e$ is constant in the edge flow, any minimizer of
$\Phi_\lambda$ is a WE under the shifted costs and can be computed with a
standard WE solver supporting additive edge tolls. Section~\ref{subsec:existence}
establishes existence, characterization, and uniqueness; we now turn to the
theoretical properties of $\Phi_\lambda$.

\section{THEORETICAL ANALYSIS}

\noindent All results are stated under one or more of the following
assumptions. \textbf{(A1) Latency regularity:} For every
$e\in\edges$, the latency function
$\ell_e:\mathbb R_{\ge0}\to\mathbb R_{\ge0}$ is continuously
differentiable and nondecreasing.
\textbf{(A2) Strict monotonicity:} For every $e\in\edges$,
$\ell_e$ is strictly increasing.

\noindent Sections~\ref{subsec:existence}--\ref{subsec:robustness} establish
existence, sensitivity, efficiency, and robustness under demand shifts.\footnote{\rev{These results hold for any fixed
$\mu\ge0$; FRC is evaluated empirically.}}
\subsection{Existence and Characterization}
\label{subsec:existence}
\begin{proposition}[Existence, characterization, and uniqueness]
\label{prop:reg_eq_basic}
Under \emph{(A1)}, for every $\lambda\ge0$:
\emph{(i)} The potential $\Phi_\lambda$ is continuous and convex on
the feasible path-flow set $\mathcal X$ and admits at least one
minimizer.
\emph{(ii)} A feasible path flow $x^\lambda$ minimizes
$\Phi_\lambda$ if and only if its induced edge flow $f^\lambda$
satisfies
$\langle
\tilde{\mathbf C}_\lambda(f^\lambda),g-f^\lambda
\rangle
\ge0,
\  g\in\mathcal E(\mathcal X).$
Equivalently, $x^\lambda$ is a Wardrop equilibrium under the shifted
edge costs
$\tilde{\ell}_e^\lambda(z)=\ell_e(z)+\lambda\mu_e$.
\emph{(iii)} Under the additional assumption \emph{(A2)}, the
induced equilibrium edge flow $f^\lambda$ is unique, although its
decomposition into path flows need not be unique.
\end{proposition}
\noindent Prop.~\ref{prop:reg_eq_basic} is the standard Beckmann
potential and variational-inequality characterization  (\cite{beckmann1956studies,facchinei2003finite}). We omit
its proof.
\subsection{Sensitivity to Regularization}
\label{subsec:sensitivity}
We next study how the equilibrium responds to changes in the
regularization parameter. For a fixed demand, let
$\tau\in\mathbb R^{|\edges|}$ be an additive edge-toll vector.
The resulting edge and path costs are
$\ell_e^\tau(z):=\ell_e(z)+\tau_e$ and
$J_p^\tau(x):=J_p(x)+\sum_{e\in p}\tau_e$, respectively.
A tolled Wardrop equilibrium has an induced edge flow $f^\tau$
satisfying
\begin{equation}
\left\langle
\mathbf C(f^\tau)+\tau,\;g-f^\tau
\right\rangle
\ge0
\qquad
\forall g\in\mathcal \mathcal E(\mathcal X).
\label{eq:vi_tau}
\end{equation}

\noindent For separable differentiable latencies, the Pigouvian toll associated
with a socially optimal flow $f^\star$ is
$\tau_e^\star=f_e^\star\ell'_e(f_e^\star)$. The proposed
toll is $\tau^\lambda:=\lambda\mu$; for brevity, write
$f^\lambda:=f^{\tau^\lambda}$ and let $f^0=f^{\mathrm{WE}}$ denote
the untolled equilibrium.

\begin{proposition}[Hölder toll sensitivity under BPR latencies]
\label{prop:toll_perturbation}
Suppose
\(
\ell_e(x)=a_e+\beta_e x^p,
\ \  x\ge0,
\)
with a common exponent $p\ge1$ and
$\beta_{\min}:=\min_{e\in\edges}\beta_e>0$.
Let $m:=|\edges|$ and define
\(
\kappa:=\beta_{\min}m^{-(p-1)/2},
\ \ 
R_{\mathrm H}(s):=\left(\frac{s}{\kappa}\right)^{1/p}.
\)
Let $f^\tau$ and $f^\sigma$ be tolled Wardrop equilibrium edge flows
for the same demand, satisfying~\eqref{eq:vi_tau}, and set
$\delta:=\|\tau-\sigma\|$. Then
$\|f^\tau-f^\sigma\|
\le R_{\mathrm H}(\delta)
=
\left(\tfrac{\|\tau-\sigma\|}{\kappa}\right)^{1/p}$.
Thus the equilibrium edge flow is Hölder continuous of order $1/p$
with respect to the toll vector. In particular, for
$\tau^\lambda=\lambda\mu$,
\begin{equation}
\scalebox{0.9}{$\displaystyle
\|f^\lambda-f^{\lambda'}\|
\le
R_{\mathrm H}\bigl(|\lambda-\lambda'|\|\mu\|\bigr),
\quad
\|f^\lambda-f^0\|
\le
R_{\mathrm H}\bigl(\lambda\|\mu\|\bigr)
$}
\label{eq:beta_holder}
\end{equation}
\end{proposition}
\begin{proof}
{Let $\Delta f:=f^\tau-f^\sigma$. Testing the VI~\eqref{eq:vi_tau} for
$f^\tau$ with $g=f^\sigma$, and the VI for $f^\sigma$ with $g=f^\tau$, gives
$\langle\mathbf C(f^\tau)+\tau,\,f^\sigma-f^\tau\rangle\ge0$ and
$\langle\mathbf C(f^\sigma)+\sigma,\,f^\tau-f^\sigma\rangle\ge0$. Adding the
two inequalities and applying Cauchy--Schwarz yields}
\begin{equation}
\left\langle\mathbf C(f^\tau)-\mathbf C(f^\sigma),\Delta f\right\rangle
\le\left\langle\sigma-\tau,\Delta f\right\rangle\le\delta\|\Delta f\|.
\label{eq:vi_comparison}
\end{equation}
{We bound the left side of~\eqref{eq:vi_comparison} from below. The
constant terms $a_e$ cancel, so
$\langle\mathbf C(f^\tau)-\mathbf C(f^\sigma),\Delta f\rangle
=\sum_{e\in\edges}\beta_e\bigl((f_e^\tau)^p-(f_e^\sigma)^p\bigr)
(f_e^\tau-f_e^\sigma)$. For $x\ge y\ge0$ and $p\ge1$, superadditivity of
$t\mapsto t^p$ on $\mathbb R_{\ge0}$ gives $x^p-y^p\ge(x-y)^p$; hence
$(x^p-y^p)(x-y)\ge|x-y|^{p+1}$ for all $x,y\ge0$, and the sum is at least
$\beta_{\min}\|\Delta f\|_{p+1}^{p+1}$. Since $p+1\ge2$, H\"older's inequality
on $\mathbb R^m$ gives $\|\Delta f\|_{p+1}\ge m^{1/(p+1)-1/2}\|\Delta f\|$,
and $\therefore\left\langle\mathbf C(f^\tau)-\mathbf C(f^\sigma),\Delta f\right\rangle
\ge\beta_{\min}m^{-(p-1)/2}\|\Delta f\|^{p+1}=\kappa\|\Delta f\|^{p+1}$.
Combining with~\eqref{eq:vi_comparison} yields
$\kappa\|\Delta f\|^{p+1}\le\delta\|\Delta f\|$. The claim is immediate if
$\Delta f=0$; otherwise, dividing by $\|\Delta f\|$ gives
$\|\Delta f\|\le(\delta/\kappa)^{1/p}=R_{\mathrm H}(\delta)$. Finally, for
$\tau=\lambda\mu$ and $\sigma=\lambda'\mu$ we have
$\delta=|\lambda-\lambda'|\,\|\mu\|$, and $\lambda'=0$ recovers $f^0$, which
gives~\eqref{eq:beta_holder}.}
\end{proof}

\noindent Thus, under standard BPR latencies, equilibrium edge flow is
$1/p$-Hölder continuous w.r.t.  changes in the toll vector.

\begin{corollary}[Lipschitz sensitivity under positive flow]
\label{cor:toll_lipschitz}
Under Proposition~\ref{prop:toll_perturbation}, suppose there exists
$\underline f>0$ such that
$f_e^\tau\ne f_e^\sigma\Rightarrow
\min\{f_e^\tau,f_e^\sigma\}\ge\underline f$.
Thus, if every edge whose flow changes remains positive and bounded away
from zero, the Hölder estimate sharpens to a Lipschitz bound. Define
$\alpha:=p\beta_{\min}\underline f^{p-1}$ and
$R_{\mathrm L}(s):=s/\alpha$. Then
\begin{equation}
\|f^\tau-f^\sigma\|
\le R_{\mathrm L}(\delta)
=\tfrac{\|\tau-\sigma\|}{\alpha}.
\label{eq:toll_lipschitz}
\end{equation}
\begin{equation}
\scalebox{0.92}{$\displaystyle
\|f^\lambda-f^{\lambda'}\|
\le R_{\mathrm L}\bigl(|\lambda-\lambda'|\|\mu\|\bigr),
\ 
\|f^\lambda-f^0\|
\le R_{\mathrm L}\bigl(\lambda\|\mu\|\bigr).
$}
\label{eq:beta_lipschitz}
\end{equation}
\end{corollary}
\begin{proof}
 For every edge on which the flows
differ, the mean-value theorem gives some $\xi_e$ between
$f_e^\tau$ and $f_e^\sigma$ such that
$\bigl(\ell_e(f_e^\tau)-\ell_e(f_e^\sigma)\bigr)
(f_e^\tau-f_e^\sigma)
=\ell'_e(\xi_e)|f_e^\tau-f_e^\sigma|^2$.
Since $\xi_e\ge\underline f$,
$\ell'_e(\xi_e)=p\beta_e\xi_e^{p-1}\ge\alpha$; edges with equal flows
contribute nothing. Hence
$\langle\mathbf C(f^\tau)-\mathbf C(f^\sigma),\Delta f\rangle
\ge\alpha\|\Delta f\|^2$.
Combining this with~\eqref{eq:vi_comparison} gives
$\alpha\|\Delta f\|^2\le\delta\|\Delta f\|$, which proves
\eqref{eq:toll_lipschitz}; substituting
$\tau=\lambda\mu$ and $\sigma=\lambda'\mu$ gives
\eqref{eq:beta_lipschitz}.
\end{proof}

\noindent We next show that increasing the curvature toll decreases the penalty-weighted aggregate flow on targeted edges.

\begin{proposition}[Penalty-weighted flow monotonicity]
\label{prop:bottleneck_monotonicity}
Under \emph{(A1)}, let $f^\lambda$ and $f^{\lambda'}$ be any
equilibrium edge flows corresponding to $\lambda$ and $\lambda'$.
Then
\begin{equation}
\lambda'\mathrel{\rev{>}}\lambda\ge0
\quad\Longrightarrow\quad
\langle\mu,f^{\lambda'}\rangle
\le
\langle\mu,f^\lambda\rangle.
\label{eq:monotonicity}
\end{equation}
\rev{If \emph{(A2)} also holds, \eqref{eq:monotonicity} extends to
$\lambda'=\lambda$ with equality.}
\end{proposition}
\begin{proof}
\rev{For $\lambda'=\lambda$ under \emph{(A2)}, Proposition~\ref{prop:reg_eq_basic}(iii)
gives $f^{\lambda'}=f^\lambda$. Let $\lambda'>\lambda$.}
The equilibrium $f^\lambda$ satisfies \eqref{eq:vi_tau} with toll
$\lambda\mu$. Using the feasible flow $f^{\lambda'}$ as the comparison
flow gives
\(
\langle
\mathbf C(f^\lambda)+\lambda\mu,\,
f^{\lambda'}-f^\lambda
\rangle
\ge0.
\)
Similarly, the equilibrium $f^{\lambda'}$ satisfies the variational
inequality with toll $\lambda'\mu$. Using $f^\lambda$ as its comparison
flow gives
\(
\langle
\mathbf C(f^{\lambda'})+\lambda'\mu,\,
f^\lambda-f^{\lambda'}
\rangle
\ge0.
\)
Adding these inequalities and rearranging yields
\(
\langle
\mathbf C(f^{\lambda'})-\mathbf C(f^\lambda),\,
f^{\lambda'}-f^\lambda
\rangle
+
(\lambda'-\lambda)
\langle
\mu,f^{\lambda'}-f^\lambda
\rangle
\le0.
\)
The first term is nonnegative by the monotonicity of the latency
operator under \emph{(A1)}. Since $\lambda'-\lambda>0$, this gives \eqref{eq:monotonicity}.
\end{proof}

\noindent This result tells us that increasing $\lambda$ weakly decreases the weighted
aggregate flow across penalized edges\rev{, although flow on an individual penalized edge may increase}.

\subsection{Efficiency Guarantees}
\label{subsec:efficiency}

\noindent We bound efficiency loss by the mismatch between
the structural toll $\lambda\mu$ and the Pigouvian toll
$\tau^\star$, using two additional assumptions.
\textbf{(A3) Smooth social cost.}
The marginal social cost varies at most linearly with the edge flow:
there exists $L_{\syscost}<\infty$ such that
$\|\nabla\syscost(f)-\nabla\syscost(g)\|
\le L_{\syscost}\|f-g\|$ for all
$f,g\in\mathcal E(\mathcal X)$. Consequently,
\(
\syscost(f)
\le
\syscost(g)
+\langle\nabla\syscost(g),f-g\rangle
+\frac{L_{\syscost}}{2}\|f-g\|^2.
\)
\footnote{For BPR latencies, \emph{(A3)} holds automatically because
$\syscost$ is a polynomial with bounded Hessian on the compact feasible
edge-flow set.}
\textbf{(A4) Componentwise toll approximation:}
For prescribed tolerances $\varepsilon_e\ge0$, there exists
$\lambda\ge0$ such that
$|\lambda\mu_e-\tau_e^\star|\le\varepsilon_e$
for every $e\in\edges$.

\begin{theorem}[Social-cost gap from toll approximation]
\label{thm:social_cost_gap}
Under \emph{(A3)} and the conditions of
Proposition~\ref{prop:toll_perturbation}, let $f^\star$ be socially
optimal and let $\tau_e^\star=f_e^\star\ell'_e(f_e^\star)$ be its
Pigouvian toll, so that $f^\star=f^{\tau^\star}$. For any toll $\tau$,
let $f^\tau$ be its equilibrium at the same demand and define
$\delta:=\|\tau-\tau^\star\|$ and
$B_{\mathrm H}(s):=sR_{\mathrm H}(s)
+(L_{\syscost}/2)R_{\mathrm H}(s)^2$, where $R_{\mathrm H}(s)=(s/\kappa)^{1/p}$. Then
\begin{equation}
\|f^\tau-f^\star\|\le R_{\mathrm H}(\delta),\ \ 
0\le\syscost(f^\tau)-\syscost(f^\star)
\le B_{\mathrm H}(\delta),
\label{eq:sc_gap}
\end{equation}
\end{theorem}
\begin{proof}
Let $\Delta f:=f^\tau-f^\star$. Since $\tau^\star$ is the Pigouvian
toll at $f^\star$,
$\nabla\syscost(f^\star)=\mathbf C(f^\star)+\tau^\star$.
Taking $g=f^\star$ in the VI for $f^\tau$ gives
$\langle\mathbf C(f^\tau)+\tau,\Delta f\rangle\le0$. Therefore,
$\langle\nabla\syscost(f^\star),\Delta f\rangle
=\langle\mathbf C(f^\tau)+\tau,\Delta f\rangle
-\langle\mathbf C(f^\tau)-\mathbf C(f^\star),\Delta f\rangle
+\langle\tau^\star-\tau,\Delta f\rangle
\le-\langle\mathbf C(f^\tau)-\mathbf C(f^\star),\Delta f\rangle
+\delta\|\Delta f\|$, using Cauchy--Schwarz.
\begin{align}
\text{From \emph{(A3)}, }&\quad
\syscost(f^\tau)-\syscost(f^\star)
\le
\delta\|\Delta f\|
\notag\\
&-\langle\mathbf C(f^\tau)-\mathbf C(f^\star),\Delta f\rangle
+\tfrac{L_{\syscost}}2\|\Delta f\|^2.
\label{eq:common_cost_bound}
\end{align}
Proposition~\ref{prop:toll_perturbation} gives
$\langle\mathbf C(f^\tau)-\mathbf C(f^\star),\Delta f\rangle
\ge\kappa\|\Delta f\|^{p+1}$ and
$\|\Delta f\|\le R_{\mathrm H}(\delta)$. Hence the negative term
in~\eqref{eq:common_cost_bound} is at most
$-\kappa\|\Delta f\|^{p+1}\le0$, and consequently
$\syscost(f^\tau)-\syscost(f^\star)
\le\delta\|\Delta f\|+(L_{\syscost}/2)\|\Delta f\|^2
\le\delta R_{\mathrm H}(\delta)
+(L_{\syscost}/2)R_{\mathrm H}(\delta)^2
=B_{\mathrm H}(\delta)$.
Since $\delta,L_{\syscost}\ge0$, substituting
$\|\Delta f\|\le R_{\mathrm H}(\delta)$ into each term gives the second inequality.
The lower bound follows from the social optimality of $f^\star$.
\end{proof}

\begin{corollary}[Lipschitz social-cost bound under positive flow]
\label{cor:lipschitz_cost_gap}
Under the conditions of Theorem~\ref{thm:social_cost_gap}, suppose the
positive-flow condition of Corollary~\ref{cor:toll_lipschitz} holds for
$f^\tau$ and $f^\star$. Define
$\delta:=\|\tau-\tau^\star\|$,
$\alpha:=p\beta_{\min}\underline f^{p-1}$,
$R_{\mathrm L}(s):=s/\alpha$, and
$\bar L_{\syscost}:=\max\{L_{\syscost},2\alpha\}$. Then $\|f^\tau-f^\star\|
\le R_{\mathrm L}(\delta)=\frac{\delta}{\alpha}$ and
\begin{equation}
0\le\syscost(f^\tau)-\syscost(f^\star)
\le
({\bar L_{\syscost}}/{2\alpha^2})\delta^2.
\label{eq:lipschitz_cost_gap}
\end{equation}

\begin{figure*}[t]
  \centering
  \includegraphics[width=\textwidth]{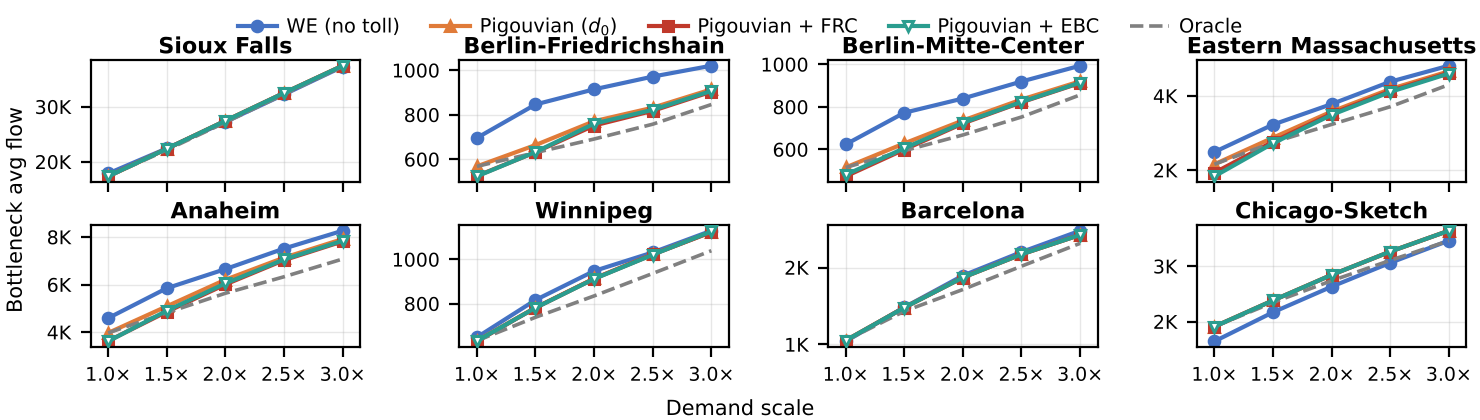}
\caption{\small \rev{Bottleneck average flow under targeted demand
scaling. Tolls are fixed at $d_0$ with the calibrated $r^\ast$; the
oracle is recomputed at each scale. Pig.+FRC tracks the oracle closely
where FRC identifies bottlenecks.}}
  \label{fig:bot_trajectories}
\end{figure*}
\begin{figure*}[t]
  \centering
  \includegraphics[width=\textwidth]{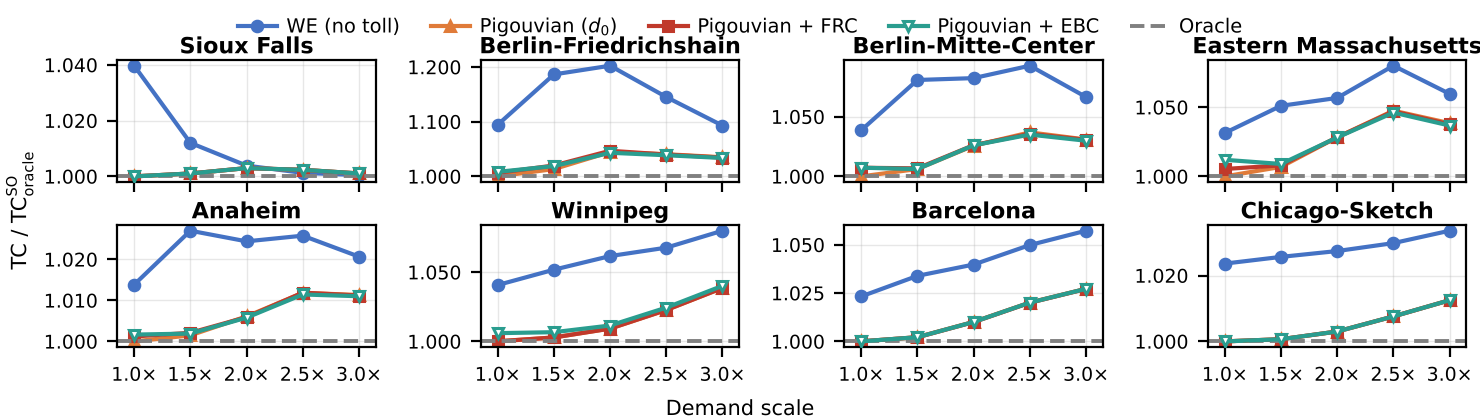}
  \caption{\small \rev{Total cost relative to oracle SO
  ($\mathrm{TC}/\mathrm{TC}^{\mathrm{SO}}_{\mathrm{oracle}}$) vs.\ targeted demand
  scale. Tolls fixed at $d_0$; oracle recomputed at each
  scale.}}
  \label{fig:tc_ratio}
\end{figure*}
\begin{figure*}[!t]
\centering
\includegraphics[width=\textwidth]{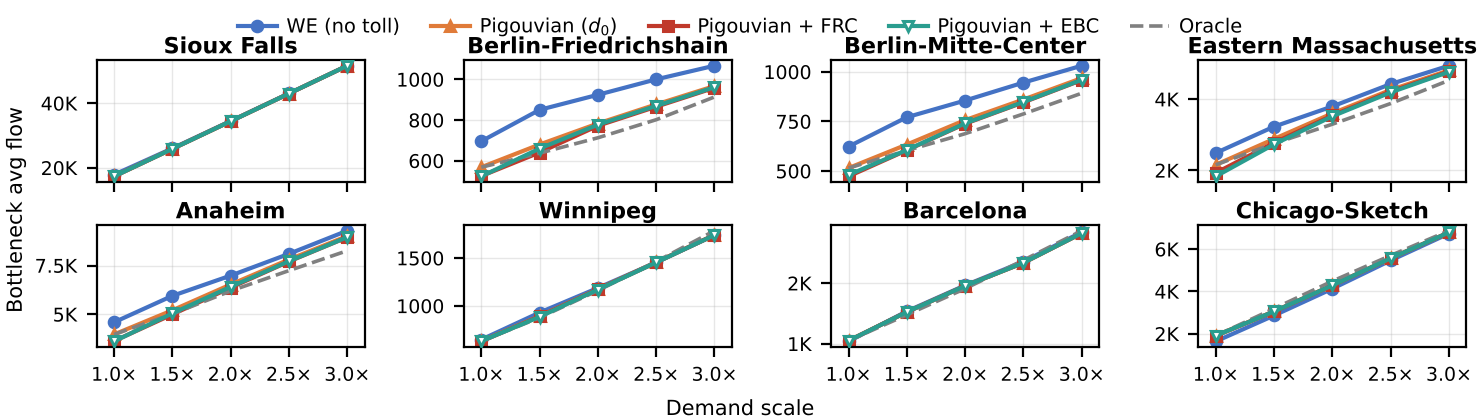}
\caption{\rev{Bottleneck average flow under \emph{uniform} demand scaling
(all OD pairs scaled equally from $1\times$ to $3\times$). The combined
toll provides bottleneck protection even when demand is not specifically
targeted at bottleneck corridors.}}
\label{fig:uniform_stress}
\end{figure*}

\end{corollary}
\begin{proof}
Corollary~\ref{cor:toll_lipschitz} gives
$\|\Delta f\|\le\delta/\alpha$. The positive-flow condition also gives
$\langle\mathbf C(f^\tau)-\mathbf C(f^\star),\Delta f\rangle
\ge\alpha\|\Delta f\|^2$. Using the smoothness constant
$\bar L_{\syscost}$ in~\eqref{eq:common_cost_bound} yields
\(
\syscost(f^\tau)-\syscost(f^\star)
\le
\delta\|\Delta f\|
+
\left(\tfrac{\bar L_{\syscost}}2-\alpha\right)\|\Delta f\|^2\) and thus ~\eqref{eq:lipschitz_cost_gap}.
\end{proof}
\begin{table}[t]
\centering
\caption{\small Computation times for FRC, EBC, and the social-optimum solve
used to compute Pigouvian tolls\rev{, and \# of bottleneck links
$|\mathcal B|$. Bottleneck links satisfy $F_e\le\bar\kappa$ and base-WE
volume/capacity $\ge\bar v$, with $(\bar\kappa,\bar v)=(-4,0.5)$ except
Barcelona $(-8,0.8)$, Winnipeg $(-5,0)$ and Chicago-Sketch $(-8,0)$%
}.}
\label{tab:settings}
\small
\setlength{\tabcolsep}{4pt}
\begin{tabular}{lrrrrr}
\toprule
\textbf{Network} & \textbf{$|\edges|$} & \rev{\textbf{$|\mathcal B|$}} & \textbf{FRC (ms)} & \textbf{EBC (ms)} & \textbf{SO (s)} \\
\midrule
Sioux Falls     &   76 & \rev{10}  & 0.1 &    1.0 &   0.5 \\
Eastern Mass.   &  258 & \rev{15}  & 0.5 &    9.3 &   2.3 \\
Berlin-Fried.   &  523 & \rev{17}  & 1.2 &   87.6 &   2.4 \\
Berlin-Mitte    &  871 & \rev{18}  & 2.2 &  287.4 &   6.7 \\
Anaheim         &  914 & \rev{29}  & 2.1 &  320.9 &   4.3 \\
Barcelona       & 2522 & \rev{84}  & 6.8 & 1978.6 &  43.1 \\
Winnipeg        & 2836 & \rev{18}  & 6.1 & 2128.5 &  56.3 \\
Chicago-Sketch  & 2950 & \rev{140} & 6.3 & 1782.5 & 124.7 \\
\bottomrule
\end{tabular}
\end{table}

\begin{corollary}[Efficiency bounds from toll approximation]
\label{cor:combined_efficiency}

Under Theorem~\ref{thm:social_cost_gap}, let
$\varepsilon:=\|\lambda\mu-\tau^\star\|$. Then
$\|f^\lambda-f^\star\|\le R(\varepsilon),\text{ and }
0\le\syscost(f^\lambda)-\syscost(f^\star)\le B(\varepsilon)$.
If \emph{(A4)} holds and
$E_\varepsilon:=(\sum_{e\in\edges}\varepsilon_e^2)^{1/2}$, the
same bounds hold with $E_\varepsilon$.
\end{corollary}
\begin{proof}
Apply Theorem~\ref{thm:social_cost_gap} with $\tau=\lambda\mu$.
Under \emph{(A4)}, $\varepsilon\le E_\varepsilon$, since  $R_{\mathrm H}$ and
$B_{\mathrm H}$ are monotone, the claim
 follows.
\end{proof}

\noindent The preceding efficiency bounds depend on the toll
approximation error $\|\regparam\penweightvec-\tau^\star\|$.
For a fixed $\penweightvec$, the following result
identifies the optimal scale $\regparam\geq0$ and expresses the
remaining error through the cosine similarity between
$\penweightvec$ and $\tau^\star$.

\begin{proposition}[Structural-toll approximation quality]
\label{prop:frc_approx}
Let $\tau^\star\ne0$ and $\mu\ne0$ be componentwise nonnegative
vectors representing, respectively, the Pigouvian toll associated with
a socially optimal flow $f^\star$ and the structural penalty. Define
$\rho:=\langle\mu,\tau^\star\rangle/
(\|\mu\|\|\tau^\star\|)$. Then
$\lambda^\star:=\langle\mu,\tau^\star\rangle/\|\mu\|^2
=\rho\|\tau^\star\|/\|\mu\|$ minimizes
$\|\lambda\mu-\tau^\star\|^2$ over $\lambda\ge0$, with
\begin{equation}
\min_{\lambda\ge0}\|\lambda\mu-\tau^\star\|^2
=(1-\rho^2)\|\tau^\star\|^2.
\label{eq:min_approx_err}
\end{equation}
Define $\delta^\star:=\sqrt{1-\rho^2}\|\tau^\star\|$. From Theorem~\ref{thm:social_cost_gap}, the equilibrium
$f^{\lambda^\star}$ under toll $\lambda^\star\mu$ satisfies
\begin{equation}
0\le
\syscost(f^{\lambda^\star})-\syscost(f^\star)
\le B_{\mathrm H}(\delta^\star).
\label{eq:rho_cost_bound}
\end{equation}
\end{proposition}
\begin{proof}
Expanding the squared error gives
$\|\lambda\mu-\tau^\star\|^2
=\lambda^2\|\mu\|^2
-2\lambda\langle\mu,\tau^\star\rangle
+\|\tau^\star\|^2$. This strictly convex quadratic is minimized at
$\lambda^\star=\langle\mu,\tau^\star\rangle/\|\mu\|^2$.
Componentwise nonnegativity gives $\lambda^\star\ge0$, so this is
the minimizer over $\lambda\ge0$. Substitution gives
$\|\lambda^\star\mu-\tau^\star\|^2
=(1-\rho^2)\|\tau^\star\|^2$, proving
\eqref{eq:min_approx_err}. 
As $\delta^\star=\|\lambda^\star\mu-\tau^\star\|$,
\eqref{eq:rho_cost_bound} follows from
Theorem~\ref{thm:social_cost_gap}
\end{proof}
\vspace{-7pt}
\noindent The quantity $\rho$ measures directional alignment between
$\mu$ and $\tau^\star$, and $\sqrt{1-\rho^2}$ is the smallest relative
toll-approximation error.\footnote{Computing $\rho$ or
$\lambda^\star$ requires $\tau^\star$ (oracle
calibration). The vector $\mu$, however, is fixed from the graph,
latencies, and base-demand AON loading.}
If the positive-flow condition of
Corollary~\ref{cor:lipschitz_cost_gap} also holds, we get
$\syscost(f^{\lambda^\star})-\syscost(f^\star)
\le\bar L_{\syscost}(1-\rho^2)\|\tau^\star\|^2/(2\alpha^2)$.
\subsection{Demand Robustness }
\label{subsec:robustness}
\noindent The preceding bounds compare a given toll with the
Pigouvian toll at the same demand. In practice, a toll calibrated
at $d_0$ may stay in use after demand shifts to $d_{new}$.
We compare the combined toll
$\tau^\star(d_0)+\lambda\mu$ with the stale toll
$\tau^\star(d_0)$ at $d_{new}$. The next theorem gives
an alignment condition under which the correction
yields a smaller certified social-cost bound.

\begin{theorem}[Demand robustness of the combined toll]
\label{thm:demand_robustness}
Under \emph{(A3)} and Proposition~\ref{prop:toll_perturbation}, let
$\tau^\star(d)$ and $f^\star(d)$ denote the Pigouvian toll and social
optimum at demand $d$. At demand $d_{new}$, let $f^c$ be the equilibrium
under $\tau^c:=\tau^\star(d_0)+\regparam\penweightvec$, and set
$\Delta\tau^\star:=\tau^\star(d_{new})-\tau^\star(d_0)$,
$\delta_c:=\|\regparam\penweightvec-\Delta\tau^\star\|$, and
$\delta_s:=\|\Delta\tau^\star\|$. Then
\begin{equation}
0\le\syscost(f^c)-\syscost(f^\star(d_{new}))
\le B_{\mathrm H}(\delta_c).
\label{eq:robustness_gap}
\end{equation}
The stale toll $\tau^\star(d_0)$ has certified bound
$B_{\mathrm H}(\delta_s)$. If
\begin{equation}
\langle\regparam\penweightvec,\Delta\tau^\star\rangle
>\tfrac12\regparam^2\|\penweightvec\|^2,
\label{eq:dominance_condition}
\end{equation}
then $B_{\mathrm H}(\delta_c)<B_{\mathrm H}(\delta_s)$, so the
combined toll has a strictly smaller certificate than the stale
toll.
\end{theorem}
\begin{proof}
At demand $d_{new}$, the combined- and stale-toll errors relative to
the current Pigouvian toll are
$\tau^c-\tau^\star(d_{new})
=\regparam\penweightvec-\Delta\tau^\star$ and
$\tau^s-\tau^\star(d_{new})=-\Delta\tau^\star$, respectively.
Applying Theorem~\ref{thm:social_cost_gap} to these two errors gives
\eqref{eq:robustness_gap}.
Moreover,
$\delta_c^2
=\delta_s^2+\regparam^2\|\penweightvec\|^2
-2\langle\regparam\penweightvec,\Delta\tau^\star\rangle$.
Thus, \eqref{eq:dominance_condition} is equivalent to
$\delta_c^2<\delta_s^2$, and hence to $\delta_c<\delta_s$.
Since $B_{\mathrm H}$ is strictly increasing, this implies
$B_{\mathrm H}(\delta_c)<B_{\mathrm H}(\delta_s)$.
\end{proof}

\noindent{If the positive-flow condition of
Corollary~\ref{cor:lipschitz_cost_gap} also holds for both comparisons,
the same argument with its quadratic bound
$\bar L_{\syscost}\delta^2/(2\alpha^2)$ in place of $B_{\mathrm H}$ gives
$\bar L_{\syscost}\delta_c^2/(2\alpha^2)<\bar L_{\syscost}\delta_s^2/(2\alpha^2)$
under~\eqref{eq:dominance_condition}.}%
\footnote{\rev{It certifies a smaller bound, not a smaller realized cost, and requires $\tau^\star(d_{new})$, so it cannot be checked at $d_0$.}}

\vspace{-1mm}
\section{EXPERIMENTAL EVALUATION}

\noindent\textbf{Setup.}
We evaluate \rev{five} toll policies: untolled Wardrop equilibrium (WE),
stale Pigouvian $\tau^\star(d_0)$, combined Pigouvian+FRC
$\tau^\star(d_0)+\lambda\mu$, \rev{combined Pigouvian+EBC} and an
oracle (SO recomputed fresh at each demand level), across eight benchmark
networks ranging from 76 to 2,950 links.\footnote{\rev{All equilibria are
computed by Frank--Wolfe to relative gap $10^{-6}$.}}
Bottleneck edges \rev{$\mathcal B$} satisfy \rev{$F_e\le\bar\kappa$ and
$\tfrac{volume}{capacity}\ge\bar v$ (Table~\ref{tab:settings})} at the base
WE, comprising 0.6--13.2\% of links depending on the network. We select
$\lambda$ by a grid search \rev{over $r=\|\lambda\mu\|/\|\tau^\star(d_0)\|$
on a single stress scenario (targeted $1.5\times$), increasing $r$ until the
bottleneck-flow deviation from the oracle stops decreasing or the cost ratio
rises by more than 1\%}, then hold it fixed across
demand scales\rev{; all other scenarios are held
out}.\footnote{\rev{$r\in\{0,0.001,\dots,2\}$. Calibrating at $d_0$, where
$\tau^\star(d_0)$ is optimal, gives $\lambda=0$.}}
As a structural baseline, we use \emph{edge betweenness
centrality} (EBC), the fraction of shortest paths passing
through each edge--a topology-derived alternative to FRC identifying
bottleneck edges.
We form a combined toll $\tau^\star(d_0) + \lambda_{\mathrm{EBC}}\mu^{\mathrm{EBC}}$
with $\mu^{\mathrm{EBC}}_e = \mathrm{EBC}(e) \cdot \ell'_e(f_e^{\mathrm{AON}})$.

\noindent\textbf{Targeted and uniform stress.}
Under targeted demand scaling, Pig+FRC reduces the
oracle-flow deviation relative to the stale Pigouvian toll \rev{by
13--20\% on Berlin-Friedrichshain (BF), Berlin-Mitte (BM), Eastern
Massachusetts (EM) and Anaheim (AN), marginally on Winnipeg (WI) and
Barcelona (BA), but not on Sioux Falls (SF) and Chicago-Sketch (CS), where
calibration selects $r^\ast=0$}
(Fig.~\ref{fig:bot_trajectories}\rev{, Table~\ref{tab:results}}). This protection
\rev{changes the
worst-case cost ratio $\eta$ by $-0.12\%$ to $+0.22\%$}
(Fig.~\ref{fig:tc_ratio}). The effect persists under uniform
scaling\rev{, with 22--29\% reductions on the same networks} (Fig.~\ref{fig:uniform_stress}).%
\footnote{\rev{As
Theorem~\ref{thm:demand_robustness} predicts, $\cos(\mu,\Delta\tau^\star)>0$
on all networks, and \eqref{eq:dominance_condition} holds at every
$r^\ast>0$ and scale.}}

\noindent\rev{\textbf{Comparison with EBC.}
Pig+EBC, calibrated by the same rule, reduces the targeted deviation on
$\mathcal B$ by 12--19\% on BF, BM, EM and AN (less than Pig+FRC on three),
selects $r^\ast=0$ on SF and CS, and attains a slightly lower $\eta$ (by
0.04--0.35\%). On $\mathcal C$ (Table~\ref{tab:results}), it reduces the
deviation more on EM, AN and WI, and Pig+FRC more on BF, BM and BA. Of the $|\mathcal B|$ links with the
largest penalty, 22--71\% (FRC) and 22--67\% (EBC) are among the
$|\mathcal B|$ most congested links at base demand.} FRC is
$8$--$347\times$ faster than edge betweenness (Table~\ref{tab:settings}). \rev{Both are
small relative to the SO solve, but EBC costs $O(|\mathcal V||\edges|)$ and
must be recomputed globally after any network change, whereas $F_e$ depends
only on links incident to $u$ and $v$, so it is linear in $|\edges|$ on
road networks and can be updated locally.}

\noindent\textbf{Conclusions.}
We proposed an FRC-based correction to Pigouvian tolls calibrated
at a base demand. Our analysis establishes conditions for efficiency
and robustness to demand shifts, while experiments show favorable
tradeoffs between total travel cost and bottleneck protection
\rev{on most networks} under targeted and uniform demand scaling. The
FRC-based toll \rev{performs comparably to EBC and} is substantially cheaper
to compute, while the fixed correction avoids recomputing tolls. Future work will explore online
calibration, broader structural guarantees, and dynamic traffic models.

\begin{table}[t]
\centering
\caption{\small \rev{Pig+FRC\,/\,Pig+EBC at the calibrated $r^\ast$ vs.\ stale
Pigouvian on held-out scenarios (targeted $2$--$3\times$, uniform
$1.5$--$3\times$): change in mean deviation from the oracle flow (\%) on
$\mathcal B$ and on the curvature-free set $\mathcal C$ of the $|\mathcal B|$
links with the largest base-WE delay ratio $\ell_e(f_e)/\ell_e(0)$; one $r^\ast$ is shown when both coincide.}}
\label{tab:results}
\footnotesize
\setlength{\tabcolsep}{2.2pt}
\begin{tabular}{lccccc}
\toprule
 & & \multicolumn{2}{c}{Deviation on $\mathcal B$} & \multicolumn{2}{c}{Deviation on $\mathcal C$} \\
\cmidrule(lr){3-4}\cmidrule(lr){5-6}
Network & $r^\ast$ & Targ. & Unif. & Targ. & Unif. \\
\midrule
Sioux Falls & 0 & \multicolumn{4}{c}{both identical to stale Pigouvian} \\
E.\ Mass. & .1/.2 & $-13$/$-19$ & $-22$/$-28$ & $-6$/$-16$ & $-12$/$-30$ \\
Berlin-F. & .05/.1 & $-20$/$-12$ & $-29$/$-16$ & $-17$/$-16$ & $-27$/$-26$ \\
Berlin-M. & .2 & $-16$/$-13$ & $-29$/$-26$ & $-23$/$-19$ & $-25$/$-25$ \\
Anaheim & .05 & $-18$/$-14$ & $-26$/$-23$ & $-9$/$-17$ & $-13$/$-26$ \\
Barcelona & .005 & $-0.1$/$-0.1$ & $+0.4$/$+0.2$ & $-2$/$-1$ & $-1$/$-0.5$ \\
Winnipeg & .02/.2 & $-0.9$/$-0.6$ & $-3$/$-9$ & $-16$/$-28$ & $-3$/$-35$ \\
Chicago & 0 & \multicolumn{4}{c}{both identical to stale Pigouvian} \\
\bottomrule
\end{tabular}
\end{table}

\bibliographystyle{IEEEtran}  
\bibliography{ieeeconf/references}

\end{document}